\documentclass[aps,prx,onecolumn,superscriptaddress,nobibnotes,nofootinbib]{revtex4-2}

\usepackage{silence}
\usepackage[linesnumbered, ruled, vlined]{algorithm2e}
\usepackage{amsmath, amssymb, amsthm, mathrsfs, amsfonts, dsfont}
\usepackage{bbm}
\usepackage{bm}
\usepackage{booktabs}
\usepackage{braket}
\usepackage{calc}
\usepackage[style=base]{caption}
\usepackage{enumerate}
\usepackage{enumitem}
\usepackage{hyperref}
\usepackage[margin=1in]{geometry}
\usepackage{graphicx}
\usepackage{ifthen}
\usepackage{lipsum}
\usepackage{makecell}
\usepackage{mathtools}
\usepackage{multirow}
\usepackage{nicematrix}
\usepackage{optidef}
\usepackage{orcidlink}
\usepackage{physics}
\usepackage{siunitx}
\usepackage{subfiles}
\usepackage{subcaption}
\usepackage{thm-restate}
\usepackage{tikz}
\usepackage{xparse}
\hypersetup{colorlinks=true,linkcolor=blue,citecolor=blue,urlcolor=blue}

\newcommand{\black}[1]{\textcolor{black}{#1}}

\newcommand{\Rank}[1]{\mathrm{rank}\qty(#1)}

\newcommand{\Cpp}{C\nolinebreak[4]\hspace{-.05em}\raisebox{.4ex}{\relsize{-3}{\textbf{++}}}}
\newcommand{\defeq}{\coloneqq}

\newcommand{\relmiddle}[1]{\mathrel{}\middle#1\mathrel{}}
\newcommand{\qBinom}[2]{\genfrac{[}{]}{0pt}{}{#1}{#2}}

\ExplSyntaxOn
\NewDocumentCommand{\definealphabet}{mmmm}{
\int_step_inline:nnn{`#3}{`#4}{
\cs_new_protected:cpx{#1 \char_generate:nn{##1}{11}}{
\exp_not:N #2{\char_generate:nn{##1}{11}}}}}
\ExplSyntaxOff

\definealphabet{bb}{\mathbb}{A}{Z}
\definealphabet{rm}{\mathrm}{A}{Z}
\definealphabet{cal}{\mathcal}{A}{Z}
\definealphabet{frak}{\mathfrak}{a}{z}

\newtheorem{theorem}{Theorem}

\newtheorem{lemma}{Lemma}

\newtheorem{corollary}{Corollary}

\allowdisplaybreaks[4]

\usetikzlibrary{
  calc,
  math,
  matrix,
  patterns,
  backgrounds,
  arrows.meta,
}

\SetKwComment{Comment}{/* }{ */}

\DontPrintSemicolon

\newboolean{isMain}
\setboolean{isMain}{true}

\begin{document}
\title{Faster Computation of Nonstabilizerness}

\author{Hiroki Hamaguchi\,\orcidlink{0009-0005-7348-1356}}
\email{hamaguchi-hiroki0510@g.ecc.u-tokyo.ac.jp}
\affiliation{Graduate School of Information Science and Technology, University of Tokyo, Tokyo, 7-3-1 Hongo, Bunkyo-ku, Tokyo 113-8656, Japan}

\author{Kou Hamada\,\orcidlink{0009-0005-9046-9818}}
\email{zkouaaa@g.ecc.u-tokyo.ac.jp}
\affiliation{Graduate School of Information Science and Technology, University of Tokyo, Tokyo, 7-3-1 Hongo, Bunkyo-ku, Tokyo 113-8656, Japan}

\author{Naoki Marumo\,\orcidlink{0000-0002-7372-4275}}
\affiliation{Graduate School of Information Science and Technology, University of Tokyo, Tokyo, 7-3-1 Hongo, Bunkyo-ku, Tokyo 113-8656, Japan}

\author{Nobuyuki Yoshioka\,\orcidlink{0000-0001-6094-8635}}
\email{nyoshioka@ap.t.u-tokyo.ac.jp}
\affiliation{Department of Applied Physics, University of Tokyo, 7-3-1 Hongo, Bunkyo-ku, Tokyo 113-8656, Japan}
\affiliation{Theoretical Quantum Physics Laboratory, RIKEN Cluster for Pioneering Research (CPR), Wako-shi, Saitama 351-0198, Japan}
\affiliation{JST, PRESTO, 4-1-8 Honcho, Kawaguchi, Saitama, 332-0012, Japan}


\begin{abstract}
    The characterization of nonstabilizerness is fruitful due to its application in gate synthesis and classical simulation.
    In particular, the resource monotone called the \textit{stabilizer extent} is a useful tool to estimate the simulation cost using rank-based simulators, one of the state-of-the-art simulators of Clifford+$T$ circuits.
    In this work, we propose \black{faster} numerical algorithms to compute the stabilizer extent.
    Our algorithm utilizes the Column Generation method, which iteratively updates the subset of pure stabilizer states used for calculation.
    This subset is selected based on the overlaps between all stabilizer states and a target state.
    In order to update the subset, we make use of a newly proposed subroutine for calculating the \textit{stabilizer fidelity} that (i) achieves linear time complexity with respect to the number of stabilizer states, (ii) super-exponentially reduces the space complexity by in-place calculation, and (iii) prunes unnecessary states for the computation.
    As a result, our algorithm can compute the stabilizer fidelity and the stabilizer extent for Haar random pure states up to $n=9$ qubits, which naively requires a memory of $\SI{305}{\exbi\byte}$.
    We further show that our algorithm runs faster when the target state vector is real. We prove that the problem size is reduced by $\order{2^n}$ compared to the general cases, which makes it computable for the case of $n=10$ qubits.
\end{abstract}
\maketitle

\section{Introduction}

In the domain of universal fault-tolerant quantum computation, elementary gates are often formulated to include both classically simulatable gates and more resource-intensive gates, as exemplified by the prominent Clifford+$T$ formalism of the magic state model~\cite{gottesmanHeisenbergRepresentationQuantum1998, nielsenQuantumComputationQuantum2010, bravyiUniversalQuantumComputation2005, litinskiGameSurfaceCodes2019, horsmanSurfaceCodeQuantum2012, fowler2019lowoverheadquantumcomputation}. Since Clifford circuits are classically simulatable~\cite{gottesmanHeisenbergRepresentationQuantum1998}, non-Clifford gates are essential for achieving quantum advantage~\cite{gidneyHowFactor20482021, leeEvenMoreEfficient2021, vonburgQuantumComputingEnhanced2021, yoshiokaHuntingQuantumclassicalCrossover2023}, and naturally, it is crucial to improve and characterize classical simulation algorithms to quantitatively understand the computational speedups in quantum circuits~\cite{bravyiTradingClassicalQuantum2016, bravyi2016improved, Bravyi2019simulationofquantum,tirritoQuantifyingNonstabilizernessEntanglement2023, hahnQuantifyingQubitMagic2022, PhysRevLett.132.240602, seddonQuantifyingQuantumSpeedups2021, liuManybodyQuantumMagic2022, leoneStabilizerEnyiEntropy2022}, as well as to understand many-body phenomena in near-Clifford circuits~\cite{zhou2020single, white2021conformal, sewell2022manaa, oliviero2022magicstatea}.

\black{
    One of the earliest efforts to quantify the simulation cost was based on the negativity~\cite{pashayanEstimatingOutcomeProbabilities2015}. Inspired by the quantum Monte Carlo techniques heavily used in quantum physics simulations~\cite{ceperley1995path, foulkes2001quantum},  Ref.~\cite{pashayanEstimatingOutcomeProbabilities2015} proposed to utilize the quasiprobability decomposition---obtained from the Wigner function for qudits with odd prime dimension---to estimate the expectation value of a physical observable.
    Later, Ref.~\cite{PhysRevLett.118.090501} presented another quasiprobability-based simulator for qubits, with costs quantified by a magic monotone known as the Robustness of Magic (RoM).
}

Alongside these {\it quasiprobability-based simulators}~\cite{pashayanEstimatingOutcomeProbabilities2015, PhysRevLett.118.090501, seddonQuantifyingMagicMultiqubit2019, seddonQuantifyingQuantumSpeedups2021}, another class of simulators, {\it rank-based simulators}~\cite{bravyiTradingClassicalQuantum2016,bravyi2016improved,Bravyi2019simulationofquantum, heimendahlStabilizerExtentNot2021, seddonQuantifyingQuantumSpeedups2021}, has emerged as a compelling alternative. A striking difference is that a rank-based simulator allows strong simulation; it can not only sample but also compute the value of amplitudes. These algorithms also exhibit quadratically improved complexity for expectation-value simulation over quasiprobability-based simulators~\cite{bravyi2016improved, seddonQuantifyingQuantumSpeedups2021}.
The simulation cost of rank-based simulators is precisely given by {\it stabilizer rank}~\cite{bravyiTradingClassicalQuantum2016}, which was later shown to be connected to another monotone called the {\it stabilizer extent}. Specifically, for approximate simulations with improved scaling relative to the original rank simulators, the cost is instead bounded by the stabilizer extent.

The definitions of monotones, such as the RoM and stabilizer extent, are given in practice by a convex optimization problem. However, the exact evaluation of these quantities is exceedingly challenging due to the super-exponential number of states involved.
Among these optimization problems, the authors showed that the Robustness of Magic (RoM), formulated via Linear Program (LP), can be computed for systems up to 8 qubits by combining the Column Generation (CG) method and \black{faster} overlap calculations~\cite{Hamaguchi2024handbookquantifying}.
Naturally, this raises the question of whether the computation of the stabilizer extent can also be accelerated.

While the formalisms may seem to resemble each other at a glance, a faster calculation in practice is not straightforward for the following reasons.
(i) The problem formulations differ.
While RoM exploits the density matrix of the target state, the stabilizer extent relies on the state vector. Thus, numerical techniques such as the Fast Walsh--Hadamard Transform, used in RoM computations, cannot be applied to the stabilizer extent calculation. (ii) The optimization problem is classified as a more challenging class of the Second-Order Cone Program (SOCP), reflecting that variables are complex rather than real.
Because of these complications, whether the stabilizer extent can be computed for larger qubit counts has remained unclear.

In this work, we show that the above-raised issues can be solved so that the computation of the stabilizer extent can be accelerated even further than the one for the RoM.
We utilize a canonical form of pure stabilizer states that allows us to enumerate them and perform \black{faster} overlap computations. When we search for the subset of stabilizer states during the CG method, we can prune unnecessary states for the calculation.
We numerically demonstrate that our proposed algorithm allows us to compute the stabilizer extent of a Haar random pure state up to $n=9$ qubits, which naively requires a memory of $\SI{305}{\exbi\byte}.$
\black{Furthermore, we show that real-amplitude state vectors, comprising GHZ (Greenberger--Horne--Zeilinger) or W states and eigenstates of physically meaningful Hamiltonians, can be computed with even less computational cost.}
Concretely, the size of the relevant stabilizer states is reduced by $\order{2^n}$ to compute the stabilizer extent of a random state up to $n=10$ qubits.

The remainder of this paper is organized as follows.
In Section~\ref{sec:preliminary}, we present the preliminaries on the formalism of the stabilizer extent.
In Section~\ref{sec:scalingUp}, we first introduce how to calculate the overlaps between all stabilizer states and a target state in Theorem~\ref{thm:complexityStabilizerOverlap}, which serves as the main subroutine for our algorithm.
Utilizing these overlap values, our proposed algorithm can compute the stabilizer extent up to 9-qubit states with drastically reduced computational resources.
In Section~\ref{sec:restrictedRealProblem}, we discuss cases where the target state possesses specific properties, mainly when the state vector is real.
We demonstrate that we can compute the stabilizer extent for up to 10-qubit states with real amplitudes.
Finally, in Section~\ref{sec:discussion}, we present our findings and offer future perspectives on our work.

\begin{table}[b]
    \caption{
        The size of $\calS_n$, the data size of $A_n$ in sparse matrix format~\cite{scipyScipySparseCsc_matrix}, and the time to compute the stabilizer extent for Haar random states by the naive method or our proposed method in Section~\ref{subsec:CG}.
        The~last~row shows the results for the case in Section~\ref{sec:restrictedRealProblem}.
    }
    \label{table:sizeOfCalSn}
    \centering
    \begin{tabular}{c|ccccccc}
        \toprule
        n             & 5                     & 6                    & 7                    & 8                    & 9                    & 10                 \\
        \midrule
        $|\calS_n|$   & 2.42e+06              & 3.15e+08             & 8.13e+10             & 4.18e+13             & 4.29e+16             & 8.79e+19           \\
        size of $A_n$ & \SI{1011}{\mebi\byte} & \SI{254}{\gibi\byte} & \SI{153}{\tebi\byte} & \SI{153}{\pebi\byte} & \SI{305}{\exbi\byte} & \SI{1}{\yobi\byte} \\
        naive         & \SI{7.7}{\minute}     & $\times$             & $\times$             & $\times$             & $\times$             & $\times$           \\
        proposed      & \SI{1.4}{\second}     & \SI{2.5}{\second}    & \SI{9.2}{\second}    & \SI{4.3}{\minute}    & \SI{7.7}{\hour}      & $\times$           \\
        real case     & \SI{0.5}{\second}     & \SI{0.5}{\second}    & \SI{3.5}{\second}    & \SI{14.7}{\second}   & \SI{6.9}{\minute}    & \SI{4.7}{\hour}    \\
        \bottomrule
    \end{tabular}
\end{table}

\section{Preliminaries}
\label{sec:preliminary}

Let $\calS_n$ be the entire set of $n$-qubit stabilizer states.
For $\ket{\phi_j} \in \calS_n$, we define the density matrix of $\ket{\phi_j}$ as $\sigma_j \defeq \ketbra{\phi_j}{\phi_j}$.
The size of $\calS_n$ scales super-exponentially as $\abs{\calS_n} = 2^n \prod_{k=1}^{n} (2^k + 1)= 2^{\order{n^2}}$ \cite[Proposition 1]{PhysRevA.70.052328}.
See also TABLE~\ref{table:sizeOfCalSn} for the size of $\calS_n$.

It is informative to provide the definition of the \textit{Robustness of Magic} (RoM), which quantifies the nonstabilizerness of an $n$-qubit mixed state $\rho$ as follows~\cite{PhysRevLett.118.090501}:
\begin{equation*}
    \calR(\rho) \defeq \min_{c \in \bbR^{\abs{\calS_n}}}
    \left\{ \norm{c}_1 \relmiddle| \rho = \sum_{j=1}^{\abs{\calS_n}} c_j \sigma_j \right\}.
\end{equation*}
On the other hand, the \textit{stabilizer extent} quantifies an $n$-qubit pure state $\ket{\psi}$ as follows~\cite[Definition 3]{Bravyi2019simulationofquantum}:
\begin{equation}\label{eq:stabilizerExtentPrimalOrig}
    \xi(\psi) \defeq \min_{c\in \bbC^{\abs{\calS_n}}}
    \left\{ \norm{c}_1^2 \relmiddle| \ket{\psi} = \sum_{j=1}^{\abs{\calS_n}} c_j \ket{\phi_j} \right\}.
\end{equation}
This definition can be simplified as the complex $L^1$-norm minimization problem as follows:
\begin{equation}\label{eq:stabilizerExtentPrimal}
    \sqrt{\xi(\psi)} = \min_{x \in \bbC^{\abs{\calS_n}}}
    \left\{ \norm{x}_1 \relmiddle| A_n x = b \right\},
\end{equation}
where $A_n \in \bbC^{2^n \times \abs{\calS_n}}$ and $b \in \bbC^{2^n}$ are defined as $(A_n)_{i,j} \defeq \braket{i}{\phi_j}$ and $b_i \defeq \braket{i}{\psi}$ using the computational basis $\qty{\ket{i}}_{i=0}^{2^n-1}$.
The problem \eqref{eq:stabilizerExtentPrimal} is pointed out to be a Second-Order Cone Program (SOCP)~\cite{heimendahlStabilizerExtentNot2021}. Thus, by defining $\calA_n$ as the column set $\{a_j\}_{j=1}^{\abs{\calS_n}}$ \black{where $a_j$ is the $j$-th column vector of $A_n$}, its dual problem can be derived as~\cite[Appendix A]{heimendahlStabilizerExtentNot2021}\cite[Section 5.1.6]{boydConvexOptimization2004}
\begin{equation}\label{eq:stabilizerExtentDual}
    \sqrt{\xi(\psi)} = \max_{y \in \bbC^{2^n}} \left\{ \Re(b^\dagger y) \relmiddle| \abs{a_j^\dagger y} \leq 1
    \text{ for all $a_j \in \calA_n$} \right\},
\end{equation}
where ${}^\dagger$ denotes the conjugate transpose, or in the case of a scalar, the complex conjugate.

Further, we introduce a function $\texttt{SolveSOCP}(\calC, b)$ to describe our algorithm in later sections with the symbol $\calC \subseteq \calA_n$ representing a column subset.
The function solves a problem that can be obtained by restricting $\calA_n$ to $\calC$, and returns the solution $x$ for the corresponding restricted primal problem~\eqref{eq:stabilizerExtentPrimal} as well as the solution $y$ for the restricted dual problem of~\eqref{eq:stabilizerExtentDual}.
In the actual implementation, this function can be realized by just solving the primal problem with SOCP solvers, such as MOSEK~\cite{mosek} or CVXPY~\cite{10.5555/2946645.3007036,agrawal2018rewriting}.

\section{Scaling Up the Exact Stabilizer Extent Calculation}
\label{sec:scalingUp}

Although both the RoM and the stabilizer extent are quantifiable through convex optimization problems, it is impractical to solve them naively for $n>5$ qubit systems due to the super-exponential growth of the number of stabilizer states $\abs{\calS_n}$, as shown in TABLE~\ref{table:sizeOfCalSn}.
Moreover, to formulate the problem in the standard form, we need at least twice the memory size of $A_n$.
For the case of RoM calculation, the authors proposed to employ a classical optimization technique known as the Column Generation (CG) method~\cite{Hamaguchi2024handbookquantifying}.
However, it is nontrivial whether the same approach could be applied to other resource measures since the structure of the matrix $A_n$ used for the calculation heavily depends on the measures.
For the stabilizer extent, in particular, the SOCP is known to be a strict extension of the LP and hence more difficult to solve in general \cite[Section 4.4.2]{boydConvexOptimization2004}.

In the following, we fill in these gaps by utilizing the specific structure of stabilizer states. We first introduce a crucial subroutine for overlap calculation in Sec.~\ref{sec:coreSubroutine}, which reduces the size of $A_n$ via the branch-and-bound method~\cite{Horst1990}, and then show how the computation can be scaled up by utilizing the CG method in Sec.~\ref{subsec:CG}.

\subsection{Core Subroutine: Calculating Stabilizer Fidelity}
\label{sec:coreSubroutine}

To utilize the CG method, we first propose a subroutine to compute the {\it stabilizer fidelity} introduced in Ref.~\cite{Bravyi2019simulationofquantum}, which is the maximal overlap between the target state and all the stabilizer states.
\black{For a pure quantum state $\ket{\psi}$ with its vector $b \in \bbC^{2^n} (b_i = \braket{i}{\psi})$,} the square root of stabilizer fidelity is defined as follows:
\begin{equation*}
    \sqrt{F(\psi)} \defeq \max_{\phi_j \in \calS_n} \abs{\braket{\phi_j}{\psi}}
    = \max_{\phi_j \in \calS_n} \abs{\sum_{i=0}^{2^n-1} \braket{\phi_j}{i}\braket{i}{\psi}}
    = \max_{a_j \in \calA_n} \abs{a_j^\dagger b}.
\end{equation*}

Note that the extension of stabilizer fidelity to mixed states has been pointed out as essential for the \black{faster} computation of RoM~\cite{Hamaguchi2024handbookquantifying}. As expected from the direct relationships between the stabilizer fidelity and the stabilizer extent~\cite{heimendahlStabilizerExtentNot2021}\cite[Proposition~2]{Bravyi2019simulationofquantum},
we find the stabilizer fidelity $F(\psi)$ is also crucial for the computation of the stabilizer extent $\xi(\psi)$.

In the following, we present how to compute the stabilizer fidelity up to 9-qubit systems.
To this end, we first introduce the following theorem useful for enumerating all the stabilizer states.
This theorem is a variant of previous work
\cite[Theorem 2]{struchalinExperimentalEstimationQuantum2021b}\cite[Section~5]{nestClassicalSimulationQuantum2010}\cite[Theorem 5.(ii)]{dehaeneCliffordGroupStabilizer2003}.
The proof is given in Appendix~\ref{app:enumeration}.
\begin{restatable}{theorem}{stabilizerStatesStandardForm}
    \label{thm:stabilizerStatesStandardForm}
    Let $\bbF_2$ be the finite field with two elements.
    For all $k \in \{1,\dots,n\}$, we define
    \begin{align*}
        \calQ_k & \defeq \left\{ Q \relmiddle| Q \in \bbF_2^{k \times k} \textrm{ is an upper triangular matrix} \right\},                                                                                                       \\
        \calR_k & \defeq \left\{ R \relmiddle| R \in \bbF_2^{n \times k} \textrm{ is a reduced column echelon form matrix with $\Rank{R}=k$} \right\},                                                                           \\
        \calT_R & \defeq \left\{ \mathrlap{\:t}\phantom{Q} \relmiddle| \mathrlap{\:t}\phantom{Q} \in \bbF_2^{n \vphantom{\times k}} \textrm{ is a representative of element in the quotient space $\bbF_2^n / \Im(R)$} \right\}.
    \end{align*}
    We also define $\calS_{n,0} \defeq \left\{ \ket{t} \relmiddle| t \in \bbF_2^{n} \right\}$ and the set of states $\calS_{n,k}$ as
    \begin{equation}\label{eq:stabilizerStateStandardForm}
        \calS_{n,k} \defeq
        \left\{ \frac{1}{2^{k/2}} \sum_{x=0}^{2^k-1}(-1)^{x^\top Q x} i^{c^\top x}\ket{R x+t} \relmiddle| Q \in \calQ_k, c \in \bbF_2^k, R \in \calR_k, t\in \calT_R \right\}.
    \end{equation}
    Then, we have $\bigcup_{k=0}^{n} \calS_{n,k} = \calS_n$, where $\calS_n$ is the entire set of $n$-qubit stabilizer states.
\end{restatable}
Here, we identify a integer $\sum_{i=0}^{n-1} x_i 2^i$ with a vector $\mqty[ x_0 \; x_1 \; \cdots \; x_{n-1} ]^\top$.
Let $\ket{\phi_j}$ be one of the stabilizer states of $k>0$ in Theorem~\ref{thm:stabilizerStatesStandardForm}, $\ket{\phi_j} = \frac{1}{2^{k/2}}\sum_{x=0}^{2^k-1} (-1)^{x^\top Q x} i^{c^\top x} \ket{Rx+t}$.
The overlap between $\ket{\psi}$ and $\ket{\phi_j}$ is
\begin{align*}
    \abs{\braket{\phi_j}{\psi}}
     & = \abs{\frac{1}{2^{k/2}} \sum_{x=0}^{2^k-1} \qty((-1)^{x^\top Q x} i^{c^\top x})^\dagger \braket{Rx+t}{\psi}}   \\
     & = \black{\abs{\sum_{x=0}^{2^k-1} \qty((-1)^{x^\top Q x} i^{c^\top x})^\dagger \qty(\frac{1}{2^{k/2}}b_{Rx+t})}} \\
     & = \black{\abs{\sum_{x=0}^{2^k-1} (-1)^{x^\top Q x} i^{c^\top x} \qty(\frac{1}{2^{k/2}}b_{Rx+t}^\dagger)}}.
\end{align*}
\black{We define $P_x \defeq (1/2^{k/2}) b_x^\dagger$.
    Recall that our goal is to compute $\max_{\phi_j \in \calS_n} \abs{\braket{\phi_j}{\psi}}$.
    Since $\max_{\phi_j \in \calS_{n,0}} \abs{\braket{\phi_j}{\psi}} = \max_{0 \leq x < 2^n} \abs{b_x}$, we can obtain
    \begin{equation}\label{eq:overlapProblemAll}
        \black{
            \max_{\phi_j \in \calS_n} \abs{\braket{\phi_j}{\psi}} =
            \max\qty(  \max_{0 \leq x < 2^n} \abs{b_x}, \quad
            \max_{1 \leq k < n}\max_{\substack{R \in \calR_k\\t \in \calT_R}}\max_{\substack{Q\in \calQ_k\\ c\in\bbF_2^k}} \qty{ \abs{\sum_{x=0}^{2^k-1} (-1)^{x^\top Q x} i^{c^\top x} P_{Rx+t}}}).
        }
    \end{equation}}

\black{Now, let us consider the following problem:
    \begin{equation}\label{eq:overlapProblem}
        \max_{Q,c} \qty{ \abs{\sum_{x=0}^{2^n-1} (-1)^{x^\top Q x} i^{c^\top x} P_{x}} }.
    \end{equation}
    For any $k,R,t$ in problem~\eqref{eq:overlapProblemAll}, by using $P'_x = P_{Rx+t} (0 \leq x \leq 2^k)$ as $P_x$ in problem~\eqref{eq:overlapProblem}, we can reduce the problem~\eqref{eq:overlapProblemAll} to~\eqref{eq:overlapProblem}.
    Thus, in the following, we only consider solving~\eqref{eq:overlapProblem}.}
If we solve~\eqref{eq:overlapProblem} naively, the time complexity is $\order{2^{n+n(n+1)/2} 2^n n^2}$, where $2^{n+n(n+1)/2}$ is the number of combinations for $(Q,c)$, $2^n$ is the number of terms in the summation, and $n^2$ is the computational cost per term.
However, we find that we can solve the problem much faster.
\begin{theorem}\label{thm:overlapProblem}
    Problem~\eqref{eq:overlapProblem} can be solved with the time complexity of $\order{2^{n+n(n+1)/2}}$ and the space complexity of $\order{2^n}$.
\end{theorem}
\begin{proof}
    First, we show that we can regard $\bbF_2$ as $\{0,1\} \subset \bbZ$ in the calculation of $Q,c$ and $x$.
    The reason why we consider this substitution is that then $i^{\alpha+\beta}=i^\alpha i^\beta$ holds true, whereas in $\bbF_2$ it does not hold since $-1=i^{1+1} \neq i^{0}=1$ with $1+1 \equiv 0 \pmod 2$.
    \black{Let $\mathbbm{1}[\cdot]$ be the indicator function, which returns 1 if the given condition is true and 0 otherwise.}
    By this substitution, the term $(-1)^{x^\top Q x}$ is invariant, while the term $i^{c^\top x}$ is multiplied by $-1$ if and only if
    \begin{equation*}
        p \equiv 2 \text{ or } 3 \pmod 4 \quad \text{where} \quad  p \defeq \sum_{i=0}^{n-1} \mathbbm{1}[c_i=1 \text{ and } x_i=1]
    \end{equation*}
    \black{Now, for $k>0$ we consider $Q' \in \{0,1\}^{k \times k}$ where $Q'_{{\lambda_1},{\lambda_2}} = \mathbbm{1}[\lambda_1< \lambda_2 \text{ and } c_{\lambda_1}=c_{\lambda_2}=1]$. Then,
    \begin{align*}
        (-1)^{x^\top Q' x}=(-1)^{\binom{p}{2}}=
        \begin{cases}
            1  & \text{if $p \equiv 0 \text{ or } 1 \pmod 4$}, \\
            -1 & \text{if $p \equiv 2 \text{ or } 3 \pmod 4$}.
        \end{cases}
    \end{align*}
    Thus, by multiplying $(-1)^{x^\top Q' x}$ to the original form~\eqref{eq:stabilizerStateStandardForm}, in other words, by identifying $Q+Q' \in \bbZ^{k \times k}$ with $Q \in \bbF_2^{k\times k}$, we indeed find that the substitution is valid.}

    Let us solve the problem~\eqref{eq:overlapProblem}. The case $n=1$ is obvious.
    For $n>1$, we define
    \begin{gather*}
        Q = \begin{bmatrix}
            Q_{00} & Q_{0}^\top   \\
            0      & \overline{Q}
        \end{bmatrix} \qty(Q_{00} \in \{0,1\}, Q_0 \in \{0,1\}^{n-1}, \overline{Q} \in \{0,1\}^{(n-1) \times (n-1)}),\\
        c = \begin{bmatrix}
            c_0 \\
            \overline{c}
        \end{bmatrix} \qty(c_0 \in \{0,1\}, \overline{c} \in \{0,1\}^{n-1}),\quad
        x = \begin{bmatrix}
            x_0 \\
            \overline{x}
        \end{bmatrix} \qty(x_0 \in \{0,1\}, \overline{x} \in \{0,1\}^{n-1}).
    \end{gather*}
    Since
    $x^\top Q x = x_0 (Q_{00}+Q_0^\top \overline{x}) + \overline{x}^\top \overline{Q} \overline{x}$
    and
    $c^\top x = c_0 x_0 + \overline{c}^\top \overline{x}$,
    we can derive that
    \begin{align}
        \sum_{x=0}^{2^n-1} (-1)^{x^\top Q x} i^{c^\top x} P_x
         & = \black{\sum_{\overline{x}=0}^{2^{n-1}-1} (-1)^{\overline{x}^\top \overline{Q} \overline{x}} i^{\overline{c}^\top \overline{x}}
        P_{2\overline{x}}+\sum_{\overline{x}=0}^{2^{n-1}-1} (-1)^{ (Q_{00}+Q_0^\top \overline{x}) + \overline{x}^\top \overline{Q} \overline{x}} i^{c_0+\overline{c}^\top \overline{x}}P_{2\overline{x}+1}}                                   \notag \\
         & = \sum_{\overline{x}=0}^{2^{n-1}-1} (-1)^{\overline{x}^\top \overline{Q} \overline{x}} i^{\overline{c}^\top \overline{x}}
        \qty(P_{2\overline{x}} + (-1)^{Q_{00}+Q_0^\top \overline{x}} i^{c_0} P_{2\overline{x}+1})                                   \notag                                                                                                           \\
         & = \sum_{\overline{x}=0}^{2^{n-1}-1} (-1)^{\overline{x}^\top \overline{Q} \overline{x}} i^{\overline{c}^\top \overline{x}} \overline{P}_{\overline{x}} \label{eq:dfsRecursion}
    \end{align}
    where $\overline{P}_{\overline{x}} \defeq P_{2\overline{x}} + (-1)^{Q_{00}+Q_0^\top \overline{x}} i^{c_0} P_{2\overline{x}+1}$.
    Since equation~\eqref{eq:dfsRecursion} has the same form as the original one, the problem~\eqref{eq:overlapProblem} can be solved by recursively computing the maximum values for all candidates of $(Q_{00}, Q_0, c_0)$, and then returning the overall maximum among them.

    We now analyze the time complexity of this recursive algorithm. There are $2^{n+1}$ possible combinations of $Q_{00}, Q_0$, and $c_0$.
    For each combination, $\overline{P}_{\overline{x}}$ for all $\overline{x}$ can be computed in $\order{n2^{n-1}}$ time.
    Hence, we can establish the recurrence relation for the time complexity:
    \begin{equation*}
        T(n) = 2^{n+1} (n2^{n-1}+T(n-1)), \quad T(1) = 4.
    \end{equation*}
    Solving this recurrence relation yields
    \begin{gather*}
        T(n) = 2^{n+\frac{n(n+1)}{2}}+ \sum_{d=2}^{n} 2^{n+\frac{n(n+1)}{2}-\frac{d(d-1)}{2}}d,\\
        \frac{T(n)}{2^{n+\frac{n(n+1)}{2}}} = 1 + \sum_{d=2}^{n} 2^{-\frac{d(d-1)}{2}} d
        \leq 1 + \sum_{d=2}^{n} 2^{-d+1} d
        = 4-(n+2)2^{-n+1} \to 4 \quad (n \to \infty).
    \end{gather*}
    Therefore, the time complexity is $\order{2^{n+n(n+1)/2}}$.
    This recursive algorithm establishes the space complexity of $\order{2^n}$ through in-place computation.
\end{proof}
\newpage

\begin{figure*}[t]
    \centering
    \includegraphics[width=\columnwidth]{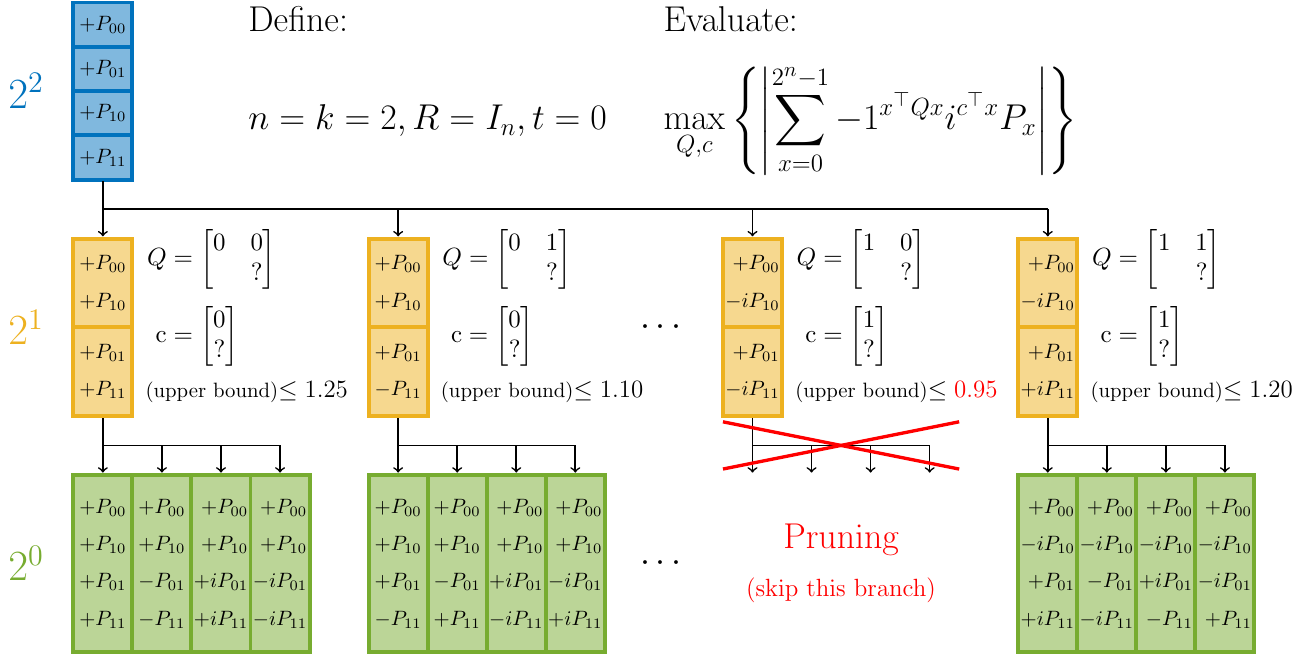}
    \caption{
        Visual explanation of stabilizer pruning.
        Each cell stores the evaluated value of the expression,
        and the $2^{n+n(n+1)/2}$ leaf nodes (green) correspond to
        the values $\sum_{x=0}^{2^n-1} (-1)^{x^\top Q x} i^{c^\top x} P_x$.
        Since we only need one set of cells per color during the procedure, we can do it in-place, and the space complexity is $\order{\sum_{i=0}^{n} 2^i}$, i.e., $\order{2^n}$.
        If the lower bound of the solution exists, for example, 1.0, we can prune a branch whose upper bound is less than 1.0.
    }
    \label{fig:BB}
\end{figure*}

Note that the actual implementation of the recursive computation in the proof of Theorem~\ref{thm:overlapProblem} is done in a slightly modified way to improve efficiency.
See our GitHub repository~\cite{Hamaguchi_stabilizer_extent_2024} for the codes, including the enumeration of the stabilizer states based on Theorem~\ref{thm:stabilizerStatesStandardForm}.

By applying the same argument for every $k,R,t$ in problem~\eqref{eq:overlapProblemAll},
we can derive the next theorem.
\begin{restatable}{theorem}{complexityStabilizerOverlap}
    \label{thm:complexityStabilizerOverlap}
    The stabilizer fidelity of a $n$-qubit pure state $\ket{\psi}$
    can be computed
    in the time complexity of
    $\order{\abs{S_n}}$ and
    the space complexity of $\order{2^n}$.
\end{restatable}

Moreover, by applying heuristic techniques to these results, we can obtain an algorithm that significantly accelerates the actual computation speed while we cannot reduce the time complexity itself.
We call this recursive algorithm to compute the stabilizer fidelity as \textit{stabilizer pruning}, which is based on Theorem~\ref{thm:complexityStabilizerOverlap} and the branch-and-bound method~\cite{Horst1990} (see FIG.~\ref{fig:BB} for its visual representation).
The main idea is that, since we are solving a maximization problem, solutions inferior to the current best solution (or to 1 for the case in Section~\ref{subsec:optimality}) are unnecessary. We can thus terminate branching if the upper bound of the current state is less than these values. For more details on the pruning strategy, see Appendix~\ref{sec:Pruning}.

Let us briefly discuss the numerical results of the stabilizer pruning.
We find that the run time to compute the stabilizer fidelity of a Haar random 8-qubit state is 5 seconds, and that of a 9-qubit state is 26 minutes.
We will use a slightly modified version of this algorithm as a subroutine to compute the stabilizer extent; namely, besides finding the maximum, we also identify other large overlaps.
All numerical experiments in this paper were conducted using \Cpp17
compiled by GCC 9.4.0 and a cluster computer powered by
Intel(R) Xeon(R) CPU E52640 v4 with \SI{270}{\giga\byte} of RAM using 40 threads.
All the codes are available at GitHub~\cite{Hamaguchi_stabilizer_extent_2024}.

\black{Comparing these results with our previous work~\cite[Corollary 1]{Hamaguchi2024handbookquantifying}, the computational complexity has been improved by a factor of $n$, and the computation time has been significantly reduced. This difference can be attributed to the fact that, while our previous work could accept general mixed state $\rho$ as input to compute the maximum fidelity, the current results only consider a pure state $\ket{\psi}$ as input to compute the stabilizer fidelity. As previously mentioned, the stabilizer fidelity for pure states can utilize the mathematically simpler representation shown in Theorem~\ref{thm:stabilizerStatesStandardForm} and stabilizer pruning, indicating a fundamental gap between the two approaches.}

\subsection{CG Method for Stabilizer Extent Calculation}
\label{subsec:CG}

Next, we introduce the Column Generation (CG) method \cite{desaulniersColumnGeneration2005}, the primary method to compute the stabilizer extent $\xi(\psi)$ up to 9-qubit systems.
Our proposed algorithm with the CG method, as outlined in Algorithm~\ref{alg:CG}, is an iterative algorithm that solves a sub-problem restricted to $\calC \subseteq \calA_n$ per iteration.
It initializes a small subset of columns $\calC_0$ and updates it so that the number of violated constraints reduces until there are no more violated columns.
The execution time of this algorithm for a Haar random state is described in TABLE~\ref{table:sizeOfCalSn}.
While we direct readers to Ref.~\cite{Hamaguchi2024handbookquantifying} for further implementation techniques, we describe two critical aspects of this algorithm: the initialization process and the solution's optimality.

\begin{algorithm}[htbp]
    \KwIn{Vector $b \in \bbC^{2^n}$ corresponding to the state $\ket{\psi}$}
    \KwOut{Exact stabilizer extent $\xi(\psi)$}
    \SetKwFunction{SolveSOCP}{SolveSOCP}
    $\calC_0 \gets \text{Partial set of $\calA_n$}$
    \Comment*[r]{Initialize using top overlap $\abs{a_j^\dagger b}$}
    \For{$k = 0, 1, 2, \ldots$} {
        $x_k, y_k \gets \SolveSOCP(\calC_k, \bm{b})$\\
        $\hat{\xi}_k(\psi) \gets \norm{x_k}_1^2$\\
        $\calC' \gets \qty{a_j \in \calA_n \relmiddle| \abs{a_j^\dagger y_k} > 1}$
        \Comment*[r]{Use of subroutine in Section~\ref{sec:coreSubroutine}}
        \If {$\calC' = \emptyset$} {
            \Return $\xi(\psi) = \hat{\xi}_k(\psi)$
        }
        $\calC_{k+1} \gets \calC_{k} \cup \calC'$
    }
    \caption{Exact Stabilizer Extent Calculation by Column Generation}
    \label{alg:CG}
\end{algorithm}

\subsubsection{Initialization}
\label{subsec:initialization}

We find that the quality of the initial guess and the number of optimization steps of Algorithm~\ref{alg:CG} can be significantly improved by choosing a subset $\calC_0 \subseteq \calA_n$ in descending order of $\abs{a_j^\dagger b}$, which can be computed with the complexity as stated in Theorem~\ref{thm:complexityStabilizerOverlap}.
While the size of $\calC_0$ is arbitrary, in computations in TABLE~\ref{table:sizeOfCalSn}, we have chosen the size to be 10,000 for $n \leq 8$ and 100,000 for $n=9$. \black{We have included all the computational bases $\calS_{n,0} = \left\{ \ket{t} \relmiddle| t \in \bbF_2^{n} \right\}$ to assure the feasibility of the problem.}

The use of $\abs{a_j^\dagger b} = \abs{\braket{\phi_j}{\psi}}$ as the indicator can be justified with various interpretations, including the direct relationships between the stabilizer fidelity and the stabilizer extent~\cite{heimendahlStabilizerExtentNot2021}\cite[Proposition~2]{Bravyi2019simulationofquantum}.
Another one is to consider it as the ``closeness'' between the states
$\ket{\phi_j}$ and $\ket{\psi}$,
which means choosing states based on their overlaps is reasonable.
The numerical experiment result in FIG.~\ref{fig:CG} also supports the effectiveness of this indicator.
For a Haar random pure 8-qubit state, even if we use as small a subset as $\abs{\calC_0} = 10^{-10} \abs{\calA_n}$, the obtained value $\hat{\xi}_0(\psi)$ closely approximated the exact value $\xi(\psi)$ and outperformed randomly selected $\calC_0$.
Such behavior was typical in other test cases, including $9$-qubit states.

\subsubsection{Solution's Optimality}
\label{subsec:optimality}

The terminate criterion for Algorithm~\ref{alg:CG} is the absence of columns that violate the dual constraints $\abs{a_j^\dagger y_k} \leq 1$, which can be checked by Theorem~\ref{thm:complexityStabilizerOverlap} as well.
The termination of the CG method indicates that the optimal dual solution $y_k$ for problem~\eqref{eq:stabilizerExtentDual}
has been found, which then implies that the primal solution $x_k$ is also optimal for problem~\eqref{eq:stabilizerExtentPrimal},
due to the strong duality of the SOCP.
Consequently, we can affirm that Algorithm~\ref{alg:CG} will compute the stabilizer extent for any state $\ket{\psi}$ once it terminates.
The convergence of the CG method was further validated through numerical experiments.
For the same 8-qubit state in Sec.~\ref{subsec:initialization}, $\max_{a_j \in \calA_n} \abs{a_j^\dagger y_k}$ reached unity after 4 iterations, indicating the discovery of the optimal solution. The method converged within just 1 or 2 iterations for different test cases.

\begin{figure}[t]
    \centering
    \includegraphics[width=\columnwidth]{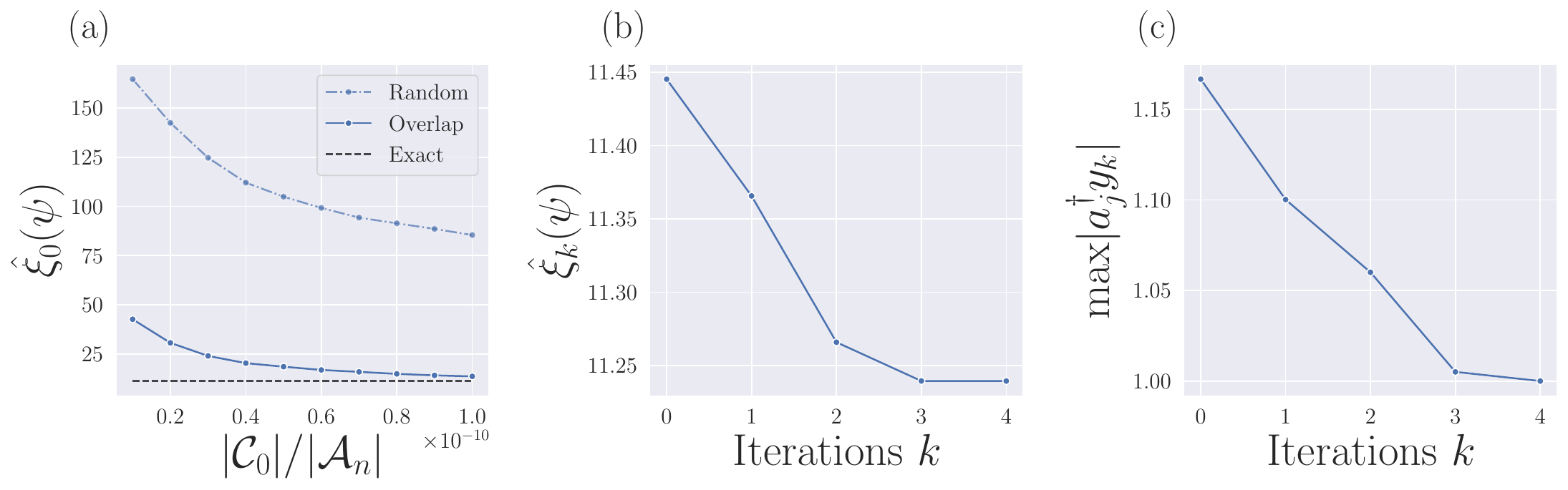}
    \caption{
        (a) Provisional value $\hat{\xi}_0(\psi)$ in Algorithm~\ref{alg:CG}
        for a Haar random 8-qubit state.
        The ratio $\abs{\calC_0}/\abs{\calA_n}$ varies from $10^{-11}$ to $10^{-10}$.
        We got much better results with the top overlap method compared to the randomly selected $\calC_0$.
        The black dashed line labeled ``Exact'' represents $\xi(\psi)$.
        (b) The convergence of the CG method for the same state.
        (c)~$\max_{a_j \in \calA_n} \abs{a_j^\dagger y_k}$ reached 1.00 after 4 iterations, indicating the optimal solution.
        The time to compute $\xi(\psi)$ was 4.3 minutes.
    }
    \label{fig:CG}
\end{figure}

\section{Calculation for States With Special Properties}
\label{sec:restrictedRealProblem}

So far, we have explored the method for calculating
the stabilizer extent applicable to the general case up to $n\leq 9$.
While the super-exponential growth of $\abs{\calS_n}$ is prohibitive, the computation can be further extended to a larger size when the target state is classified into specific classes.

One such example is a product state, $\ket{\psi} = \bigotimes_j \ket{\psi_j}$.
It is known that if all the components $\ket{\psi_j}$ are at most 3-qubit state, then the multiplicativity holds as $\xi(\psi) = \prod_j \xi(\psi_j)$~\cite{Bravyi2019simulationofquantum}.
While we cannot guarantee the multiplicativity if the factors contain a 4 or more qubits state~\cite {heimendahlStabilizerExtentNot2021}, the tensor product of each solution $x_j$ of $\xi(\psi_j)$ is still practically useful since it can be used as the initial guess for the solution of $\xi(\psi)$.

Another remarkable class of states is those with real amplitudes.
To the best of our knowledge, this property has not been investigated in previous works and offers significant calculation advantages.
One of the well-known examples are the W-state and the GHZ-state, defined as follows:
\begin{equation*}
    \ket{W} \defeq \frac{1}{\sqrt{n}} \qty(\ket{100\dots0}+\ket{010\dots0}+\dots+\ket{000\dots1}), \quad \ket{GHZ} \defeq \frac{\ket{0}^{\otimes n} + \ket{1}^{\otimes n}}{\sqrt{2}}.
\end{equation*}
Other critical applications include eigenstates of quantum many-body Hamiltonians with time-reversal symmetry, whose matrix elements are given by real components.
For instance, physically meaningful Hamiltonians such as in the XXZ and Heisenberg spin models, the transverse-field Ising model, the Fermi--Hubbard model, and the $t$-$J$ model, all preserve the time-reversal symmetry regardless of the underlying lattice. Beyond condensed matter systems, we may also consider first-principle quantum chemistry Hamiltonians or lattice gauge theory Hamiltonians such as the one-dimensional Schwinger model. Note that the abundance of examples reflects that the microscopic equation of motion is time-reversal symmetric unless there is spontaneous symmetry breaking or the weak interaction.

We envision that other symmetries contribute to accelerating the computation, while we leave the study of general symmetry as an open question.

\subsection{Reduction of Problem Size in Real-Amplitude States}

In the following, we argue that the stabilizer extent for real-amplitude states can be computed significantly faster than complex-amplitude ones, with a numerical demonstration for uniformly random quantum states with real amplitudes up to 10-qubit systems.
Targeting random states allows us to avoid assuming
unnecessary specificity, thereby demonstrating
the broad computational potential.

First, we define $\calS'_n \subset \calS_n$ as follows:
\begin{equation*}
    \calS'_n = \qty{\ket{\phi_j} \in \calS_n \relmiddle| \braket{i}{\phi_j} \in \bbR \text{ for all $i$}}.
\end{equation*}
This means that $\calS'_n$ is the union of $\calS_{n,0}$ and the set of states with $c=0$ in Theorem~\ref{thm:stabilizerStatesStandardForm}.
Since $\abs{\calS'_n} = 2^n \prod_{k=0}^{n-1} (2^k+1) = \frac{2}{2^n+1}\abs{\calS_n}$, the size is reduced by $\order{2^n}$.
Let $\calA'_n$ denote the corresponding subset of
the columns in $\calA_n$.
Also refer to FIG.~\ref{fig:Amat} for the definition of $\calA'_n$.
Then, the next lemma holds.
\begin{figure}[tbp]
    \centering
    \includegraphics[width=\columnwidth]{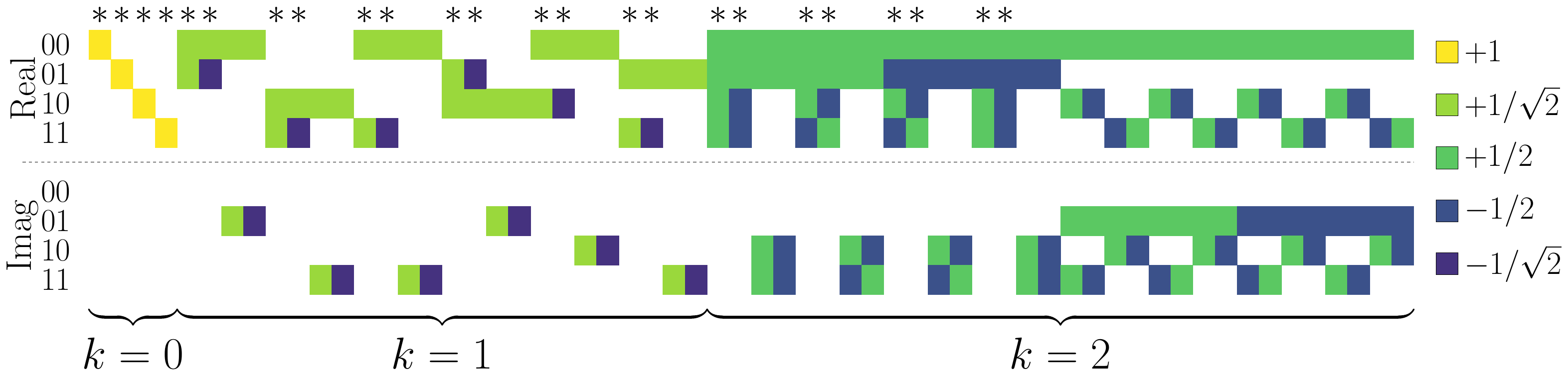}
    \caption{
        Visualization of the matrix $A_n$, i.e., the column set $\calA_n$, with $n=2$.
        The upper half corresponds to the real part, and the lower half corresponds to the imaginary part.
        The $j$-th column represents the state vector $a_j$ of the stabilizer state $\ket{\phi_j}$.
        The integer $k$ below corresponds to the integer $k$ in Theorem~\ref{thm:stabilizerStatesStandardForm}.
        By~restricting~the column set $\calA_n$ to the starred columns, which are real vectors, we can obtain $\calA'_n$.
    }
    \label{fig:Amat}
\end{figure}
\begin{lemma}{\label{lem:absRealForTIsAbsRealForS}}
    For a pure state $\ket{\psi}$, suppose \black{its vector $b \in \bbC^{2^n} \ (b_i = \braket{i}{\psi})$} is real, i.e. $b\in \bbR^{2^n}$.
    Then, we have
    \begin{equation*}
        \max_{a\in \calA_n} \abs{a^\dagger b} = \max_{a\in \calA'_n} \abs{a^\dagger b}.
    \end{equation*}
\end{lemma}
\begin{proof}
    We only show that $\max_{a\in \calA_n} \abs{a^\dagger b} \leq \max_{a\in \calA'_n} \abs{a^\dagger b}$ since $\calA_n \supset \calA'_n$.
    Fix $a \in \calA_n \setminus \calA'_n$.
    Since $a \not\in \calA'_n$, $k>0$ holds true in the form of Theorem~\ref{thm:stabilizerStatesStandardForm} and we can denote the state corresponding to $a$ as
    $\ket{\phi} = \frac{1}{2^{k/2}}\sum_{x=0}^{2^k-1} (-1)^{x^\top Q x} i^{c^\top x} \ket{Rx+t}$.
    The following two states
    \begin{equation*}
        \ket{\phi_+} \defeq \frac{1}{2^{k/2}}\sum_{x=0}^{2^k-1} (-1)^{x^\top Q x} \ket{Rx+t}, \quad
        \ket{\phi_-} \defeq \frac{1}{2^{k/2}}\sum_{x=0}^{2^k-1} (-1)^{x^\top Q x + c^\top x} \ket{Rx+t}
    \end{equation*}
    have only real amplitudes and belong to $\calS'_n$.
    Define $a_+$ and $a_-$ as the column vectors in $\calA'_n$ of $\ket{\phi_+}$ and $\ket{\phi_-}$, respectively.
    Denote $a^\dagger b = \alpha + i\beta \; (\alpha, \beta \in \bbR)$.
    Since $c\in \bbF_2^k$, we have $a_+^\dagger b = \alpha+\beta, a_-^\dagger b = \alpha-\beta$, and thus
    \begin{equation*}
        \abs{a^\dagger b}
        =\sqrt{\alpha^2 + \beta^2}
        \leq \abs{\alpha}+\abs{\beta}
        = \max\qty(\abs{\alpha+\beta}, \abs{\alpha-\beta})
        = \max\qty(a_+^\dagger b, a_-^\dagger b)
        \leq \max_{a \in \calA'_n} \abs{a^\dagger b}.
    \end{equation*}
    This completes the proof.
\end{proof}
We can derive the following corollary by combining Lemma~\ref{lem:absRealForTIsAbsRealForS} and Theorem~\ref{thm:complexityStabilizerOverlap}.
\begin{corollary}
    The stabilizer fidelity of an $n$-qubit state $\ket{\psi}$ with real amplitudes can be computed in time complexity of $\order{\abs{S_n}/2^n}$ and space complexity of $\order{2^n}$.
\end{corollary}

Now, we can prove the following theorem by Lemma~\ref{lem:absRealForTIsAbsRealForS}.
\begin{theorem}
    \label{thm:restrictedRealProblem}
    Suppose that $\ket{\psi}$ is a state
    with real amplitudes.
    The optimal solution for the restricted problem~\eqref{eq:stabilizerExtentDual} where $\calA_n$ is substituted by $\calA'_n$
    is also optimal for the original problem~\eqref{eq:stabilizerExtentDual}.
\end{theorem}
\begin{proof}
    Let $x^*$ and $y^*$ be the optimal solutions
    of the restricted primal and dual problems,
    namely, problems~\eqref{eq:stabilizerExtentPrimal}
    and~\eqref{eq:stabilizerExtentDual}
    with $\calA'_n$ instead of $\calA_n$.
    We can ensure the solutions always exist.
    Define~$x^{**} \in \bbC^{\abs{\calS_n}}$~as
    \begin{equation*}
        x^{**}_j=\begin{cases}
            x^*_j & \text{if $a_j \in \calA'_n$,}    \\
            0     & \text{if $a_j \notin \calA'_n$.}
        \end{cases}
    \end{equation*}
    We will show that $x^{**}$ and $y^{*}$ are optimal solutions for the original problems.

    Let $\mathrm{OPT}$ be the optimal value for the original problems.
    Since $x^{**}$ is a feasible solution for the original primal problem,
    it is clear that $\mathrm{OPT} \leq \norm{x^{**}}_1=\norm{x^*}_1$.
    By the strong duality theorem,
    $\mathrm{OPT}$ is also the optimal value
    for the original dual problem.
    From Lemma~\ref{lem:absRealForTIsAbsRealForS}, we know that $y^*$ is a feasible solution not only for the restricted dual problem but also for the original dual problem, hence $\mathrm{OPT} \geq \Re(b^\dagger y^*)$.
    Again, by applying the strong duality theorem
    to the restricted problems,
    we obtain $\norm{x^*}_1 = \Re(b^\dagger y^*)$,
    which implies that $\mathrm{OPT} = \norm{x^{**}}_1 = \Re(b^\dagger y^*)$.
    Therefore, $x^{**}$ and $y^*$ are optimal solutions
    for the original problems.
\end{proof}

Thanks to Theorem~\ref{thm:restrictedRealProblem}, we can reduce the column set size by $\order{2^n}$.
We also conducted a numerical experiment for a uniformly random 10-qubit state with real amplitudes.
The algorithm converged within a single iteration using $\abs{\calC_0}=100,000$ as in Section~\ref{subsec:initialization}, which was sufficiently large to obtain the stabilizer extent.
This result confirmed our success in computing its stabilizer extent, accomplished within a time frame of 4.7 hours.

\section{Discussion}
\label{sec:discussion}

In this paper, we have shown that the stabilizer fidelity and the stabilizer extent can be calculated by utilizing the specific structure of stabilizer states.
\black{One of the significant advancements is the proposal of a novel canonical form of stabilizer states, which allowed us to compute the stabilizer fidelity with time complexity of $\mathcal{O}(|S_n|)$. The central technique is the \textit{stabilizer pruning}, which skips a significant amount of calculation based on the branch structure enforced by the canonical form. Note that the complexity achieves exponential improvement per a stabilizer state compared to naive calculation and linear improvement per a stabilizer state compared to the state-of-the-art in Ref.~\cite{Hamaguchi2024handbookquantifying}.
    With the faster calculation of stabilizer fidelity, we have employed the Column Generation method to effectively choose stabilizer states that contribute to the solution of the optimization problem.
    We have numerically demonstrated the validity of our proposal up to sufficiently large systems of $n=9$ qubits. Furthermore, we have proposed a specialized algorithm for states with real amplitudes, which computes the exact value of stabilizer extent up to $n=10$ qubits in a few hours.}

A wealth of future work can be envisioned. First, it is intriguing to see whether the methodology proposed in this work generalizes to other monotones, such as the dyadic negativity~\cite{seddonQuantifyingQuantumSpeedups2021}.
Second, it is crucial to incorporate the effect of symmetry present in the target quantum state.
In particular, it is known for the computation of the RoM that permutation symmetry of quantum states reduces the size of  $\mathcal{A}_n$ super-polynomially, allowing the computation to scale up significantly larger up to 26 qubits~\cite{heinrichRobustnessMagicSymmetries2019}. Such a reduction can be expected for the case of stabilizer extent as well, while it is nontrivial how to speed up further using the techniques such as stabilizer pruning.
Third, faster evaluation of magic resources on hardware is a viable direction~\cite{oliviero2022measuring, haug2023scalable}. While recent work has shown that it is intractable to discriminate an ensemble of states with an extensive amount of nonstabilizerness from that with logarithmic amount~\cite{gu2024pseudomagic}, it is interesting to develop a method to assess whether a specific state is non-magic, i.e., constant amount of nonstabilizerness, or not.

\section*{Acknowledgments}

We thank Shigeo Hakkaku, Bartosz Regula, and Ryuji Takagi for helpful discussions.
N.Y. wishes to thank JST PRESTO for Grant No.\ JPMJPR2119 and is grateful for support from IBM Quantum. This work was supported by JST Grant No.\ JPMJPF2221, JST ERATO Grant No.\ JPMJER2302 and JST CREST Grant No.\ JPMJCR23I4, Japan.

\bibliography{stabilizerExtent,stabilizerExtent2}

\appendix
\section{Appendix for Stabilizer Pruning}

\subsection{Enumeration of Stabilizer States}
\label{app:enumeration}

In this section, we prove Theorem~\ref{thm:stabilizerStatesStandardForm}.
\stabilizerStatesStandardForm*
\begin{proof}
    The main idea comes from Ref.~\cite{struchalinExperimentalEstimationQuantum2021b}.
    From previous work \cite[Theorem 2]{struchalinExperimentalEstimationQuantum2021b}\cite[Section 5]{nestClassicalSimulationQuantum2010}\cite[Theorem 5.(ii)]{dehaeneCliffordGroupStabilizer2003}, we know that any state in $\bigcup_{k=0}^{n} \calS_{n,k}$ is a stabilizer state. Thus, we can construct an inclusion map from $\bigcup_{k=0}^{n} \calS_{n,k}$ to $\calS_n$.
    In this proof, we will show that this map is bijective, which means this map is an identity mapping.
    The assertion is trivial for the case $k=0$ with $2^n$ instances.
    We will only consider the case $k>0$.
    Define $f\colon (Q,c,R,t) \mapsto \ket{\phi}$
    as a map from $(Q,c,R,t)$ to
    the corresponding stabilizer state $\ket{\phi}$.
    We will confirm that $f$ is bijective.
    First, we show that $f$ is injective.
    We can say that
    \begin{align*}
             & \left\{R_1 x + t_1  \relmiddle| x \in \bbF_2^{k} \right\} = \left\{R_2 x + t_2  \relmiddle| x \in \bbF_2^{k} \right\} \\
        \iff & \qty(\Im(R_1)=\Im(R_2)) \land \qty(t_1-t_2 \in \Im(R_1))                                                              \\
        \iff & R_1 = R_2 \land t_1 = t_2.
    \end{align*}
    The last equivalence is due to the
    property of the reduced column echelon form and the quotient space.
    Given that $Q$ is an upper triangular matrix,
    both $Q$ and $c$ can be uniquely reconstructed from the amplitude of the state. Consequently, the function $f$ is injective.

    Second, we show that $f$ is surjective.
    Since $f$ is injective, we only have to show that the cardinality of the domain is equal to that of the codomain, i.e., $-2^n+\abs{\calS_n}$.
    It is known that the number of $\bbF_2^{n \times k}$
    reduced column echelon form matrices $R$
    with $\Rank{R}=k$ is $\qBinom{n}{k}_2$,
    which is a $q$-binomial coefficient with $q=2$~\cite[Theorem 7.1]{kacQuantumCalculus2002}.
    Therefore, the number of $Q,c,R,t$ is
    $2^{k(k+1)/2},2^k,\qBinom{n}{k}_2,2^{n-k}$, respectively.
    The total number of states is
    \begin{equation*}
        \sum_{k=1}^{n} 2^{k(k+1)/2} 2^k \qBinom{n}{k}_2 2^{n-k} \\
        = -2^n + 2^n \sum_{k=0}^{n} \qBinom{n}{k}_2 2^{k(k+1)/2}               \\
        = -2^n + 2^n \prod_{k=1}^{n} (2^k+1)
        = -2^n + \abs{\calS_n}.
    \end{equation*}
    In the second to last equation, we used the $q$-binomial theorem.
    Therefore, the mapping is surjective, which completes the proof.
\end{proof}

\subsection{Pruning for the Branch-and-Bound Method}
\label{sec:Pruning}

In Section~\ref{sec:coreSubroutine}, we explained the branching in the branch-and-bound method.
This algorithm can be much faster
by the bounding introduced in this section.
First, recall that we are maximizing the following:
\begin{equation*}
    \max_{Q,c} \qty{ \abs{\sum_{x=0}^{2^n-1} (-1)^{x^\top Q x} i^{c^\top x} P_x} }.
\end{equation*}
This can be easily bounded by
\begin{equation*}
    \max_{Q,c} \qty{ \abs{\sum_{x=0}^{2^n-1} (-1)^{x^\top Q x} i^{c^\top x} P_x} }
    \leq \max_{Q,c} \qty{ \sum_{x=0}^{2^n-1} \abs{(-1)^{x^\top Q x} i^{c^\top x} P_x} }
    = \sum_{x=0}^{2^n-1} \abs{P_x}.
\end{equation*}
Such a bound is crucial for the branch-and-bound method,
because it allows us to terminate the branch if the bound is inferior to the current best value.
However, the bound can be more refined.
Since each coefficient takes only
$1, -1, i$ or $-i$, we can bound as
\begin{equation}\label{eq:branchCutNewProblemDefinition}
    \max_{Q,c} \qty{ \abs{\sum_{x=0}^{2^n-1} (-1)^{x^\top Q x} i^{c^\top x} P_x} }
    \leq \max_{C} \abs{\sum_{x=0}^{2^n-1} i^{c_x} P_x} ,
\end{equation}
where $C=(c_x)_{x=0}^{2^n-1}$ takes values independently from the set $\{0, 1, 2, 3\}$.
Let $P_C \defeq \sum_{x=0}^{2^n-1} i^{c_x} P_x$,
and define $\theta_x \defeq \arg(i^{c_x}P_x)$.
The right-hand side of \eqref{eq:branchCutNewProblemDefinition} is equal to
\begin{equation}\label{eq:absPIsSumOfPxCos}
    \max_{C} \abs{P_C}
    = \max_{C, \theta} \qty{ \langle P_C, e^{i \theta} \rangle }
    = \max_{C, \theta} \qty{\sum_{x=0}^{2^n-1} \langle i^{c_x} P_x, e^{i \theta} \rangle}
    = \max_{C, \theta} \qty{\sum_{x=0}^{2^n-1} \abs{P_x} \cos(\theta_x -\theta)}
\end{equation}
where $\langle \cdot, \cdot \rangle$ denotes the inner product of complex values.
Then, we can confirm that Algorithm \ref{alg:branchCut} is
certain to return the optimal solution for~\eqref{eq:absPIsSumOfPxCos} as follows.
If we fix the value of $\theta$ in \eqref{eq:absPIsSumOfPxCos}, the optimal values of $C$ can be determined so that $\cos(\theta_x - \theta)$ is maximized, i.e., $\theta_x \in [\theta - \pi / 4, \theta + \pi / 4)$.
Thus, instead of trying all $\theta$, we run Algorithm \ref{alg:branchCut} to cover all the possible optimal solutions of $C$, which is sufficient to calculate $\max_{C} \abs{P_C}$.

\begin{algorithm}[ht]
    \caption{Bounding for the Branch-and-Bound Method}
    \label{alg:branchCut}
    \KwIn{Coefficients $P_x$ for $x=0,1,\dots,2^n-1$}
    \KwOut{The answer for problem~\eqref{eq:absPIsSumOfPxCos}}
    Take $C=(c_x)_{x=0}^{2^n-1}$ that satisfies $\theta_x
        \in [0, \pi / 2)$.\;
    Sort and relabel $P_x$ so that $0 \leq \theta_0 \leq \theta_1 \leq \dots \leq \theta_{2^n-1} < \pi/2$.\;
    $\mathrm{ans} \leftarrow 0$\;
    \For{$x \leftarrow 0$ \KwTo $2^n - 1$}{
        $\mathrm{ans} \leftarrow \max\qty(\mathrm{ans}, \abs{\sum_{x=0}^{2^n-1} i^{c_x} P_x})$\;
        $c_x \leftarrow c_x + 1 \pmod 4$
    }
    \KwRet{$\mathrm{ans}$}
\end{algorithm}

\begin{figure}[t]
    \centering
    \includegraphics[width=\columnwidth]{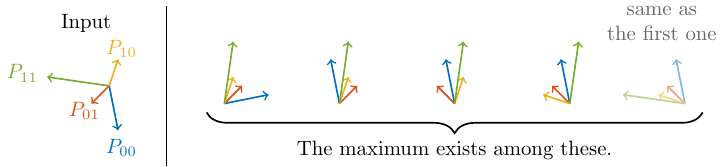}
    \caption{
        Visualization of Algorithm~\ref{alg:branchCut}.
        Suppose that $n=2$ and $P_x$ are represented as vectors in the complex plane
        (e.g., $P_{00} = 1-5i$) on the left side.
        Iterating the loop in Algorithm~\ref{alg:branchCut}
        yields $2^n$ patterns of the coefficients $c_x$,
        as depicted on the right side.
        The maximum of problem~\eqref{eq:branchCutNewProblemDefinition}
        exists among these $2^n$ patterns.
    }
    \label{fig:argsort}
\end{figure}

Refer to FIG.~\ref{fig:argsort} for a visual representation of this algorithm.
The time complexity of this approach is $\order{n2^n}$ owing to the sorting of $2^n$ elements.

As the end of this section, we evaluate the performance of this bound.
We can obtain the lower bound of \eqref{eq:absPIsSumOfPxCos} by taking the expected value with respect to $\theta$ as follows:
\begin{align*}
    \max_{C} \abs{P_C}
     & = \max_{C, \theta} \qty{\sum_{x=0}^{2^n-1} \abs{P_x} \cos(\theta_x -\theta)}             \\
     & \ge \bbE \qty[\max_{C} \qty{\sum_{x=0}^{2^n-1} \abs{P_x} \cos(\theta_x -\theta)}]        \\
     & = \sum_{x=0}^{2^n-1} \abs{P_x} \cdot \bbE \qty[\max_{c_x} \qty{\cos(\theta_x -\theta)}].
\end{align*}
Here, we assume $\theta$ is drawn from the uniform distribution over $[0, 2 \pi)$. Then, we can replace each term $\bbE \qty[\max_{c_x} \qty{\cos(\theta_x -\theta)}]$ with $\bbE\qty[\cos(\theta'_x)]$
where $\theta'_x$ follows the uniform distribution over the interval $[-\pi/4, +\pi/4)$.
Then, we can derive that
\begin{equation*}
    \frac{\max_{C} \abs{P_C}}{\sum_{x=0}^{2^n-1} \abs{P_x}}
    \ge
    \frac{\sum_{x=0}^{2^n-1} \abs{P_x} \cdot \bbE\qty[\cos(\theta'_x)]}{\sum_{x=0}^{2^n-1} \abs{P_x}}
    =\frac{\int_{-\frac{\pi}{4}}^{+\frac{\pi}{4}} \cos(\theta) \dd{\theta}}{\pi/2}
    = \frac{2\sqrt{2}}{\pi}=0.900316\cdots.
\end{equation*}
The result of a numerical experiment suggests that this lower bound serves as a rough approximation of the ratio.
We independently sampled $P_x$ from the standard normal distribution and $\theta_x$ from the uniform distribution over $[0,2\pi)$.
The numerical experiment results obtained are as follows.
First, the average of
\begin{equation*}
    \frac{\max_{C} \abs{P_C}}{\sum_{x=0}^{2^n-1} \abs{P_x}}
    = \frac{\max_{C} \qty{ \abs{\sum_{x=0}^{2^n-1} i^{c_x} P_x} }}{\sum_{x=0}^{2^n-1} \abs{P_x}}
\end{equation*}
over 100 runs yielded 0.935624 for $n=4$.
This confirms that Algorithm~\ref{alg:branchCut} provides
a better upper bound compared to $\sum_{x=0}^{2^n-1} \abs{P_x}$.
However, in the same setting, it turned out that the average of
\begin{equation*}
    \frac{\max_{Q,c} \qty{ \abs{\sum_{x=0}^{2^n-1} (-1)^{x^\top Q x} i^{c^\top x} P_x} }}{\sum_{x=0}^{2^n-1} \abs{P_x}}
\end{equation*}
yielded 0.824056, implying the
bound~\eqref{eq:branchCutNewProblemDefinition}
may not necessarily be optimal.
Whether a better bound can be obtained
with lower computational cost is
left as an open problem.

\end{document}